\documentclass[a4paper,12pt]{article}
\usepackage[bbgreekl]{mathbbol}
\PassOptionsToPackage{dvipsnames}{xcolor}
\usepackage{tcolorbox}
\usepackage[small]{titlesec}
\usepackage[utf8]{inputenc}
\usepackage[T1]{fontenc}  
\usepackage[]{mdframed}
\usepackage[backend=biber,style=alphabetic,natbib=false,giveninits=true,doi=true,isbn=false,url=false,date=year,maxbibnames=99,sorting=ynt,defernumbers=true]{biblatex}
\DeclareNameAlias{default}{family-given}

\DeclareFieldFormat[article]{title}{#1}
\renewbibmacro{in:}{}
\DeclareFieldFormat[article]{pages}{#1}
\DeclareFieldFormat{journal}{#1}
\renewcommand\bf\bfseries

\renewbibmacro*{volume+number+eid}{%
	\printfield{volume}%
	\setunit*{\addnbspace}
	\printfield{number}%
	\setunit{\addcomma\space}%
	\printfield{eid}}
\DeclareFieldFormat[article]{number}{\mkbibparens{#1}}
\DeclareFieldFormat[article]{volume}{\textbf{#1}}
\DeclareFieldFormat{year}{\mkbibparens{#1}}
\DeclareBibliographyDriver{article}{%
	\printnames{author}:%
	\newunit\newblock
	\newunit\newblock
	\printfield{journaltitle}
	\newunit
	\iffieldundef{number}{\printfield{volume}}{\printfield{volume}\addspace\printfield{number}}
	\addcomma\addspace\printfield{pages}\addspace
	\printfield{year}
}
\AtEveryBibitem{\clearfield{month}}
\AtEveryBibitem{\clearfield{day}}
\usepackage[english]{babel}
\usepackage[T1]{fontenc}
\usepackage{csquotes}
\usepackage{bbm}
\usepackage[leqno]{amsmath}
\usepackage{amsfonts,amsthm,amsbsy,amssymb,dsfont,stmaryrd}
\usepackage{braket}

\makeatletter
\newcommand{\leqnomode}{\tagsleft@true\let\veqno\@@leqno}
\newcommand{\reqnomode}{\tagsleft@false\let\veqno\@@eqno}
\makeatother

\usepackage{mathtools}
\usepackage[makeroom]{cancel}

\numberwithin{equation}{section}

\newcommand\myshade{85}
\colorlet{mylinkcolor}{violet}
\colorlet{mycitecolor}{YellowOrange}
\colorlet{myurlcolor}{Aquamarine}

\usepackage[left=2.5cm,right=2.5cm,top=2.5cm,bottom=2.5cm]{geometry}
\usepackage[unicode=true,pdfusetitle,
bookmarks=true,bookmarksnumbered=false,bookmarksopen=false,
breaklinks=false,pdfborder={0 0 1},backref=false,
linkcolor  = mylinkcolor!\myshade!black,
citecolor  = mycitecolor!\myshade!black,
urlcolor   = myurlcolor!\myshade!black,
colorlinks = true,
]
{hyperref}
\usepackage[nameinlink]{cleveref}

\usepackage{tikz}
\usetikzlibrary{arrows}
\usetikzlibrary{intersections}

\usepackage{graphicx}
\usepackage{caption}
\usepackage{slashed}
\definecolor{ct_black}{HTML}{000000}
\definecolor{ct_orange}{HTML}{ED872D}
\definecolor{ct_purple}{HTML}{7A68A6}
\definecolor{ct_blue}{HTML}{348ABD}
\definecolor{ct_turquoise}{HTML}{188487}
\definecolor{ct_red}{HTML}{E32636}
\definecolor{ct_pink}{HTML}{CF4457}
\definecolor{ct_green}{HTML}{467821}

\definecolor{ct2_green}{HTML}{9FF781}
\definecolor{ct2_green_dark}{HTML}{088A08}

\theoremstyle{plain}
\newtheorem{thm}{\protect\theoremname}[section]
\theoremstyle{plain}
\newtheorem{lem}[thm]{\protect\lemmaname}
\theoremstyle{plain}
\newtheorem{cor}[thm]{\protect\corollaryname}
\theoremstyle{plain}
\newtheorem{prop}[thm]{\protect\propositionname}
\theoremstyle{plain}

\newtheorem{assumption}[thm]{\protect\assumptionname}

\theoremstyle{remark}
\newtheorem{rem}[thm]{\protect\remarkname}

\theoremstyle{definition}
\newtheorem{defn}[thm]{\protect\definitionname}
\theoremstyle{plain}

\providecommand{\assumptionname}{Assumption}
\providecommand{\conventionname}{Convention}
\providecommand{\claimname}{Claim}
\providecommand{\corollaryname}{Corollary}
\providecommand{\definitionname}{Definition}
\providecommand{\lemmaname}{Lemma}
\providecommand{\propositionname}{Proposition}
\providecommand{\remarkname}{Remark}
\providecommand{\theoremname}{Theorem}
\providecommand{\examplename}{Example}

\crefname{section}{Section}{Sections}
\crefname{appendix}{Appendix}{Appendices}
\crefname{figure}{Figure}{Figures}
\crefname{assumption}{Assumption}{Assumptions}
\crefname{thm}{Theorem}{Theorems}
\crefname{lem}{Lemma}{Lemmas}
\crefname{rem}{Remark}{Remarks}
\crefformat{equation}{(#2#1#3)}
\crefname{table}{Table}{Tables}

\crefrangelabelformat{equation}{(#3#1#4--#5#2#6)}

\crefmultiformat{equation}{(#2#1#3}{, #2#1#3)}{#2#1#3}{#2#1#3}
\Crefmultiformat{equation}{(#2#1#3}{, #2#1#3)}{#2#1#3}{#2#1#3}

\newtheorem*{lem*}{\protect\lemmaname}

\newcommand{\ee}{\operatorname{e}}
\newcommand{\ii}{\operatorname{i}}

\newcommand{\RR}{\mathbb{R}}
\newcommand{\CC}{\mathbb{C}}

\newcommand{\GG}{\mathbb{G}}

\newcommand{\calV}{\mathcal{V}}
\newcommand{\calU}{\mathcal{U}}

\newcommand{\ti}[1]{\widetilde{#1}}

\newcommand\norm[1]{\left\lVert#1\right\rVert}

\newcommand\abs[1]{\left|#1\right|}

\newcommand{\ip}[2]{\langle #1, #2 \rangle}

\newcommand{\dif}{\operatorname{d}}

\newcommand{\szpan}{\operatorname{span}}

\renewcommand{\Re}[1]{\operatorname{\mathbb{R}\mathbbm{e}}\left\{#1\right\}}
\newcommand{\ve}{\varepsilon}
\newcommand{\vf}{\varphi}
\newcommand{\Id}{\mathds{1}}

\newcommand{\dist}{\mathrm{dist}}

\newcommand{\supp}{\operatorname{supp}}

\newcommand{\im}{\operatorname{im}}

\usepackage{environ}

\NewEnviron{malign}{%
	\begin{align}\begin{split}
			\BODY
	\end{split}\end{align}
}

\newcommand{\eq}[1]{\begin{align*}#1\end{align*}}
\newcommand{\eql}[1]{\begin{align}#1\end{align}}

\newcommand{\br}[1]{\left(#1\right)}

\title{Lower Bounds on Quantum Tunneling\\ for Excited States}
\author{\href{cf@math.princeton.edu}{Charles L. Fefferman}, \href{mailto:jacobshapiro@princeton.edu}{Jacob Shapiro}\\
	{\footnotesize Department of Mathematics, Princeton University}\\
	 \href{miw2103@columbia.edu}{Michael I. Weinstein}\\
		\footnotesize{Department of Applied Physics and Applied Mathematics,}\\\footnotesize{and Department of Mathematics, Columbia University}
}

\begin{document}
	\reqnomode
	
	\maketitle
	\begin{abstract}
		We revisit the problem of quantum tunneling for a particle moving in the continuum, and in the absence of a magnetic field. In all spatial dimensions, we extend previous results to the case where the single-well potential satisfies reflection-symmetry.
  
	\end{abstract}

    \section{Introduction}\label{sec:intro}

    We study double-well Hamiltonians on $L^2(\RR^\nu)$ of the form \eql{ H_{\lambda,d} = P^2+\lambda^2 V_d(X) } where $V_d:\RR^\nu\to\RR$ is a double-well potential (with separation $d$ between the wells) and $\lambda$ is the well-depth; $P\equiv-\ii\nabla$ is the momentum operator and $X$ is the position operator.

We are interested in double-well potentials of the following form: let $a>0$, $d\in\RR^\nu\setminus\Set{0}$. For convenience and without loss we choose $d$ along the positive $1$-axis: $\norm{d}=\ip{d}{e_1}$. We loosely refer  to both the vector $d e_1$ and its magnitude by $d$. Then we take \eql{V_d(x) := v(x) + v(x-d)\qquad (x\in\RR^\nu)}
where $v:\RR^\nu\to\RR$ is a single-atom potential with $\supp(v)\subseteq B_a(0)$ and $d>2a$ so that  \eql{\supp(v)\cap\supp(v(\cdot-d))=\varnothing\,.} Generally speaking we need at least $v\in C^2$ for most of our results.

The single well Hamiltonian (independent of $d$) is given by
\eql{ \label{eq:single well Hamiltonian}h_\lambda := P^2+\lambda^2 v(X)}
and we assume $v$ is a trapping potential so that it has discrete spectrum of bound states below its essential spectrum comprised of the Laplacian's scattering states. We denote the sequence of single well bound states and energies by $\Set{\vf_j}_{j\geq1}\subseteq L^2(\RR^\nu)$ and $\Set{e_j}_{j\geq1}\subseteq(-\infty,0)$,  respectively. Both depend on $\lambda$, but we typically suppress this dependence if there's no ambiguity. For the most part we \emph{do not} require $v$ have a unique minimum (because we don't make use of the harmonic oscillator approximation). For the remainder of this introduction we assume that the sequence of eigenvalues $\Set{e_j}_{j\geq1}$ is simple.  

For fixed $d$ and $\lambda$ sufficiently large, it was shown in \cite{Simon_1983_AIHPA_1983__38_3_295_0},  in the above setting, that the double-well Hamiltonian $H_{\lambda,d}$  also has a sequence of discrete eigenvalues below its essential spectrum, which we denote by $\Set{E_j^\pm}_{j\geq1}$. For fixed $j\geq1$, each pair of levels $E_j^\pm$ should be understood as the splitting of the single-well level $e_j$ due to quantum mechanical tunneling between the two wells. Our main object of interest is thus the $j$th splitting \eql{ \Delta_j(\lambda,d) := E_j^+-E_j^-\qquad(j\geq 1)\,.} 

The quantity $\Delta_1(\lambda,d)$ is very well-studied.  Under the assumption that $v$ is smooth and has a unique non-degenerate minimum, one has the asymptotic formula \cite{Simon_1984_10.2307/2007072,Helffer_Sjostrand_1984} \eq{
-\frac{1}{\lambda}\log\br{\Delta_1(\lambda,d) } \to S_E(0,d),\qquad {\rm as}\quad \lambda\to\infty,
} where $S_E$ is the classical Euclidean action from the point $0$ to the point $d$:  \eq{
S_E(x,y) := \inf_{\substack{
    T > 0; \\
    \gamma:[0,T] \to \mathbb{R}^\nu; \\
    \gamma(0) = x,\ \  \gamma(T) = y
}} \int_{0}^T \left[\norm{\dot{\gamma}(t)}^2+V(\gamma(t))\right]\dif{t}\,,
}  i.e., the value of the classical action, with the sign of the potential flipped (hence the name Euclidean), on the extremizing solution of Newton's equations \eq{
\ddot{\gamma } = \nabla V \circ \gamma\,.
} with appropriate boundary conditions; it is also equal to the Agmon distance \cite{Simon_1984_10.2307/2007072}. With no assumption on a unique non-degenerate minimum, it was shown in \cite{FLW17_doi:10.1002/cpa.21735} that for fixed $d$ and all $\lambda$ sufficiently large, \eq{
\Delta_1(\lambda,d) \gtrsim \exp\br{-\lambda C_{d,v}}
} for some constant $C_{d,v}>0$, which is independent of $\lambda$.

\begin{tcolorbox}The basic question we wish to address in this note, with partial answer in the case of reflection symmetric $v$, is:\\
Are there lower bounds on  $\Delta_j(\lambda,d)$ for excited states, $j\ge2$?
\end{tcolorbox}

An emergent quantity, instrumental in the study of $\Delta_j(\lambda,d)$, is the so-called \emph{hopping coefficient}; see, for example, \cite{Dimassi2010-hq,FSW_22_doi:10.1137/21M1429412,ShapWein22,SlaterKoster1954,Ashcroft_Mermin_1976}\footnote{In the semi-classical asymptotic analysis literature, it is sometimes referred to as the interaction matrix element; in chemistry it is the overlap integral.}. 
 It is defined as \eql{ \label{eq:hopping coefficient def}\rho_j \equiv \rho_j(\lambda,d):= \ip{\vf_j}{\left(H_\lambda-e_j\Id\right) R^d \vf_j } = \lambda^2\ip{\vf_j}{v(X)R^d\vf_j }\qquad(j\geq1)} where for any $u\in \RR^\nu$, $R^u$ is the shift operator, given by \eql{(R^u f)(x) \equiv f(x-u)\qquad(x\in\RR^\nu)\,.}

Actually the hopping coefficient $\rho$ is of independent interest even before studying $\Delta$; see, for example \cite{ShapWein22}, where it is central to establishing the tight-binding reduction for magnetic quantum Hall systems, and  \cite{FLW17_doi:10.1002/cpa.21735}, where the validity of Wallace's tight binding model of graphene was established in the strong binding regime.

The connection between$\Delta_j(\lambda,d)$ and $\rho_j(\lambda,d)$ is well-known; see e.g. \cite{FSW_22_doi:10.1137/21M1429412}. There it was established that \eql{\lim_{\lambda\to\infty}\frac{\Delta_1(\lambda,d)}{2|\rho_1(\lambda,d)|}=1\,.} Here we rather show that for $\lambda$ sufficiently large and fixed, \eql{\label{eq:connection between rho and Delta}\lim_{d\to\infty}\frac{\Delta_j(\lambda,d)}{2|\rho_j(\lambda,d)|}=1\qquad(j\geq1)\,.}

\subsection*{Outline} 
This paper is organized as follows. \cref{sec:1dwarmup} is a warmup in the spatial dimension $\nu=1$ case. We derive an exact formula for $\rho_j(d)$, for any $j\ge1$. This result makes use of an identity for $\rho$, which only involves information about $\vf_j$ at the bisector $x=d/2$, which is  outside the support of $v$. In \cref{sec:general-formula} we  derive a generalization, in which $\rho$ is expressed as an integral of $\vf_j$ and its derivative over the hyperplane $x_1=d/2$, which is outside the support of $v$. Although, for $j\ge2$, $\vf_j$ changes sign, this integral representation of $\rho_j$, together with the additional assumption that $v$ is reflection symmetric,  yields 
 an expression for $\rho_j$ having an definite sign. A classical result on the non-vanishing of solutions to elliptic PDE on open sets, implies a lower bound on $\rho$ as a function of $d$. In \cref{app:excited} we derive, using energy estimates proved in \cref{sec:energy-estimates}, under the assumption that $e_j$ is a simple eigenvalue, that \cref{eq:connection between rho and Delta} holds.

  \subsection*
{Acknowledgements}
MIW was supported in part by NSF grant DMS-1908657, DMS-1937254 and Simons Foundation Math + X Investigator Award \# 376319 (MIW). Part of this research was carried out during the 2023-24 academic year, when MIW was a Visiting Member in the School of Mathematics - Institute of Advanced Study, Princeton, supported by the Charles Simonyi Endowment, and a Visiting Fellow in the Department of Mathematics at Princeton University.

 \section{Warm up--the one-dimensional case}\label{sec:1dwarmup}
    It is instructive to obtain a precise expression for $\rho$ in one space dimension. Remarkably, in this one-dimensional setting, there is an explicit expression for it.
    \begin{thm}
        Let $d> 2a$. For any $j\geq1$ and $\lambda>0$,  $\rho_j$ as defined in \cref{eq:hopping coefficient def} is given by \eq{\rho_j(d) = C_j(\lambda,v)  \sqrt{-e_j} \ee^{-\sqrt{-e_j}d}} where $C_j(\lambda,v)\in\CC\setminus\Set{0}$: this constant depends on $\lambda$ and $v$, but  not on $d$.  
    \end{thm}

    \begin{proof}

    For convenience, since $j$ is fixed in this discussion, we drop it from the expressions during the proof. Hence, $\vf\equiv\vf_j$ etc.
    
    We define $\psi:=\vf(\cdot-d)$ and $c:=\frac{d}{2}$ for convenience, and make the following chain of observations:
    \eq{ \rho(d) &:= \ip{\vf}{\lambda^2 v \psi}_{L^2(\RR)}\\ 
    &=\int_{\RR} \overline{\vf} \lambda^2 v \psi \\
    &=\int_{(-\infty,c]} \overline{\vf} \lambda^2 v \psi\tag{$\supp(v)\subseteq[-a,a]$ \& $a<c<d-a$} \\
    &= \int_{(-\infty,c]} \left(\overline{e \vf + \vf''}\right)  \psi\qquad\tag{$\lambda^2v \vf = e\vf+\vf''$}\\
    &\stackrel{\text{IBP}}{=} e\ip{\vf}{\psi}_{L^2((-\infty,c])}+\overline{\vf'(c)}\psi(c)-\overline{\vf'(-\infty)}\psi(-\infty)-\int_{(-\infty,c]}\vf'\psi'\\
    &\stackrel{\text{IBP}}{=} e\ip{\vf}{\psi}_{L^2((-\infty,c])}+\overline{\vf'(c)}\psi(c)-\\
    &\qquad -\overline{\vf(c)}\psi'(c)-\overline{\vf(-\infty)}\psi'(-\infty)+\int_{(-\infty,c]}\overline{\vf}\psi''\\
    &= e\ip{\vf}{\psi}_{L^2((-\infty,c])}+\overline{\vf'(c)}\psi(c) -\overline{\vf(c)}\psi'(c)+\ip{\vf}{\psi''}_{L^2((-\infty,c])}
    } Here we used the fact that thanks to $\vf\in L^2$, the terms at $-\infty$ vanish.
    
    Now we remark that the function $\psi(x)=\varphi(x-d)$ also obeys the Schr\"odinger equation: 
    \eq{-\psi''(x)+\lambda^2 v(x-d)\psi(x)=e\psi(x).} 
    Since $d>2a$, $v(x-d)$ which is supported in the interval $[d-a,d+a]$, vanishes for $x\le c=d/2$. Hence,  \eq{-\psi''(x)=e\psi(x)\qquad {\rm for}\quad x\leq c\,. } 
    It follows that $\ip{\vf}{\psi''}_{L^2((-\infty,c])} = -e \ip{\vf}{\psi}_{L^2((-\infty,c])}$ and we find 
    \begin{equation} \rho(d) = \left(\overline{\vf'}\psi-\overline{\vf}\psi'\right)(c)\ =\ \overline{\vf'(c)}\vf(-c)-\overline{\vf(c)}\vf'(-c).
    \label{eq:rho-d1}\end{equation}
  Since $a<c=d/2<d-a$, the expression for the hopping coefficient $\rho(d)$, displayed in \cref{eq:rho-d1}, depends 
    only on the values of $\varphi$ and $\varphi'$
     at a point outside the support of $v$, where there is an explicit exponential form for the wave-function. Hence, it is simple to evaluate the expression for $\rho(d)$ in \cref{eq:rho-d1}, up to a non-vanishing $d-$ independent factor. We proceed with this evaluation.
     
  Setting  $\kappa:=\sqrt{-e}>0$, we have
  \eq{
  \vf(x)=\begin{cases}
      A_- e^{\kappa x}, & x\le-a\\
      A_+ e^{-\kappa x}, & x\ge a.
  \end{cases}
  }
    where the  constants $A_\pm$ are both non-vanishing, and  depend on $v$ and $\lambda$, but not on $d$.
   That neither $A_\pm$ do not vanish follows from global existence and uniqueness  on $\RR$ for the ODE  $h_\lambda\vf=e\vf$.

   Since $c=d/2>a$, we have \eq{ \vf(c) &= A_+\ee^{-\kappa c},\quad 
    \vf'(c) = -A_+\kappa \ee^{-\kappa c}\\
     \vf(-c) &= A_- \ee^{-\kappa c},\quad 
       \vf'(-c) = A_- \kappa\ee^{-\kappa c} 
    }

 Substitution into \cref{eq:rho-d1}, we obtain 
 \begin{equation} \rho(d)  = -\overline{A_+}\kappa e^{-\kappa c}A_-e^{-\kappa c} -\overline{A_+}e^{-\kappa c}A_-\kappa e^{-\kappa c} = -2\overline{A_+}A_-\kappa e^{-2\kappa c} = 
 -2\overline{A_+}A_-\kappa e^{-\kappa d}.
    \label{eq:rho-d2}\end{equation}
Setting $C=-2\overline{A_+}A_-\kappa$, which depends on $\lambda$, $v$ but not on $d$
completes the proof.
    \end{proof}

Note that if $v(x)=v(-x)$ then $A_-=\pm A_+$ (eigenstates are either even or odd) and we have
\eq{ \rho(d) = \mp 2|A_+|^2\kappa e^{-\kappa d}.}

    \section{A general formula for the hopping in two and higher dimensions}\label{sec:general-formula}
    \color{black}
	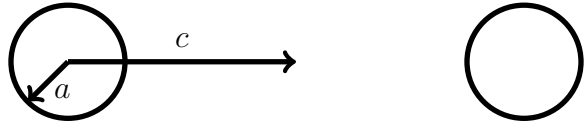
\begin{figure}
	    \centering
	    \begin{tikzpicture}[x=1cm,y=1cm,scale=0.75]
\draw [line width=2pt] (0,0) circle (1cm);
\draw [line width=2pt] (8,0) circle (1cm);
\draw [->,line width=2pt] (0,0) -- (4,0);
\draw [->,line width=2pt] (0,0) -- (-0.7,-0.7);
\node[below] at (-0.1,-0.2) {$a$};
\node[below] at (2,0.7) {$c$};

\end{tikzpicture}
	    \caption{Setup for double-well. The support radius of each well is $a$ and the half-distance between the centers of the wells is $c:=d/2$. We assume $\norm{c}>a$ so that the wells do not overlap.}
	    \label{fig:setup for double well}
	\end{figure}

We would like to generalize the argument of the previous section to spatial dimensions $\nu\ge2$. Recall $d \equiv \norm{d} \equiv \norm{d} e_1$.

    We shall show that 
    \begin{prop}
         Let $j\geq1$, assume $d > 2a$ and let $D\subseteq\RR^\nu$ be a subset with smooth boundary such that \eq{\supp(v)\subseteq\operatorname{interior}(D)\ \ {\rm and}\ \  \supp(v(\cdot-d))\subseteq\operatorname{interior}(D^c)\,.} Then, if $n:\partial D\to\RR^\nu$ is the normal unit vector pointing out of $D$, we have \eql{
        \rho_j(\lambda,d)=-\int_{\partial D} \left((\nabla \overline{\vf_j})R^d \vf_j-\overline{\vf_j}\nabla R^d \vf_j\right) \cdot n \,.\label{eq:rhod} 
         }

         In particular, with the choice \eq{D := \Set{x_1 \leq c}\,,\qquad \partial D = \Set{x_1 = c}\,,\qquad n=e_1} which satisfies the assumptions above, we get
        \eql{\label{eq:master formula for rho}\rho_j(\lambda,d)=\int_{\Set{x\in\RR^\nu|x_1=d/2}} \left[(\partial_1 \overline{\vf})R^d\vf_j-\overline{\vf}\partial_1 R^d\vf_j\right]\,.}
    \end{prop}

        \begin{rem}The formula continues to hold in the magnetic case, if one assumes that $A:\RR^\nu\to\RR^\nu$, the magnetic vector potential, is such that   \eql{\label{eq:assumption on A}[P_i,A_i(X)]=0\qquad(i=1,\dots,\nu)\,.} Moreover, $R^d\vf_j$ should now be defined using \emph{magnetic} translations.
       \end{rem}
        \begin{rem}
            The formula \cref{eq:master formula for rho} is not new. We traced its earliest appearance, or at least in similar guise, to \cite[Problem 50.3]{landau1977quantum}; see the expression for their $E_1-E_0$. See also \cite{Helffer_Sjostrand_1984} and \cite[Eq-n (4.1)]{fournais2025purely}.
        \end{rem}
        
    \begin{proof}

        We first note two identities that follow from Leibniz rule: \eql{\label{eq:IBP1}f\operatorname{div} V = \operatorname{div} (fV)-V\cdot\operatorname{grad}f} and \eql{\label{eq:IBP2}V\cdot\operatorname{grad}f = \operatorname{div} (fV) - f\operatorname{div} V}
        Integrating over $D\subseteq\RR^\nu$, and  integration by parts gives:
        \eq{ \ip{\Delta f}{g}_{L^2(D)} &\equiv \int_{D} \overline{\Delta f}g \\
        &= \int_{D} \overline{(\operatorname{div}\circ\operatorname{grad} f)}g \\
        &= \int_{D} \left[\operatorname{div}(g\operatorname{grad}\overline{f})-(\operatorname{grad}\overline{f})\cdot (\operatorname{grad}g)\right]\qquad\qquad \qquad\qquad\qquad\qquad \qquad\qquad\tag{Using \cref{eq:IBP1} with $V=\operatorname{grad}\overline{f}$} \\
        &= \int_{D} \left[\operatorname{div}(g\operatorname{grad}\overline{f})-\operatorname{div}(\overline{f}\operatorname{grad}g)+\overline{f}\operatorname{div}\circ\operatorname{grad}g\right] \tag{Using \cref{eq:IBP2} with $V=\operatorname{grad}g$}\\
        &= \ip{f}{\Delta g}_{L^2(D)} + \int_{D}\operatorname{div}\left((\nabla \overline{f})g-\overline{f}\nabla g\right) \\
        &= \ip{f}{\Delta g}_{L^2(D)} + \int_{\partial D} \left((\nabla \overline{f})g-\overline{f}\nabla g\right) \cdot n\ .
        }
        Here,  $n$ denotes the unit normal on $\partial D$, which points exterior to $ D$.

        Again as for the 1D case, for convenience we denote $c:=d/2$ and will eventually integrate over the hyperplane $x_1=c$. For convenience we define $\psi:=R^d\vf_j$ and also drop the $j$ subscripts since $j$ is fixed.

        Let now $D\subseteq \RR^\nu$ be any reasonable subset such that $\supp(v)\subseteq \operatorname{interior}(D)$ and $\supp v(\cdot-d)\subseteq \operatorname{interior}(\RR^\nu\setminus D)$. Then for $\rho$ we find 
        \eq{ \rho(d) &\equiv \ip{\vf}{\lambda^2 v \psi}_{L^2(\RR^\nu)}\\
        &= \int_{D} \overline{\vf}\lambda^2 v \psi = \int_{D} \overline{\lambda^2 v\vf} \psi \\
        &= \int_{D} \overline{\left(e-P^2\right)\vf} \psi \\
        &= e\ip{\vf}{\psi}_{L^2(D)} - \ip{P^2 \vf}{\psi}_{L^2(D)}\,.
        } 
        Now using the various integration by parts formulae, we get 
        \eq{\rho(d) &= e\ip{\vf}{\psi}_{L^2(D)} - \ip{\vf}{P^2 \psi}_{L^2(D)} -\int_{\partial D} \left((\nabla \overline{\vf})\psi-\overline{\vf}\nabla \psi\right) \cdot n \,.
        }
        A few observations are in order. First, since $\psi\equiv R^d \vf$, since we are integrating within $D$ which does not intersect $\supp(v(\cdot-d))$, on $D$, \eq{P^2\psi &\stackrel{\mathrm{on}\,D}{=}\left[P^2+\lambda^2 v(X-d)\right]\psi\\
        &=R^d \left[P^2+\lambda^2 v(X)\right]\vf\\
        &=R^d e \vf = e \psi } so that \cref{eq:rhod} follows.
    \end{proof}

    \section{Reflection symmetric potentials }

    Suppose  $v$ is reflection symmetric about the bisecting plane, $x_1=c = d/2$, between the two wells.  
     That is, $Uv=v$, where 
     $U: L^2(\RR^\nu)\to L^2(\RR^\nu)$ is defined by:
    \eq{(U f)(x_1,x_\perp) := f(2c-x_1,x_\perp), \qquad f\in L^2(\RR^\nu)\ .}

    For $j=1$, the ground state $\vf_1$ enjoys the same symmetry. Indeed, we know that $\vf_1$ is both non-degenerate and positive, and $U \vf_1$ is also an eigenstate of the Hamiltonian $h$, so we must have $U \vf_1 = \vf_1$ which is the desired symmetry. However, for higher $j$, $\vf_j$ may fail to obey the same symmetry: $\vf_j$ is not necessarily positive nor non-degenerate. However, thanks to $[h,U]=0$, the eigenspace of $h$ breaks into a direct sum of even and odd functions ($\ker\br{U\pm\Id}$) and so one may \emph{choose} all eigenfunctions of $h$ to have definite parity. We assume this is indeed the case and hence our result applies to all eigenfunctions of $h$. 

    We then see that if $U\vf_j = \pm \vf_j$, \cref{eq:master formula for rho} reduces to \eql{\label{eq:rho with reflection symmetry}\rho(d) = \pm \left.\partial_{x_1}\right|_{x_1=\frac{d}{2}}\int_{x^\perp\in\RR^{\nu-1}}\left|\vf_j(x_1,x^\perp)\right|^2\dif{x^\perp}\,.}

    The following theorem estimates the dependence of $\rho_j$ on the distance $d$, and thus may be useful in the fixed $\lambda$, large $d$ regime:
    \begin{thm}\label{thm:lb_rho_j}
        Let $\lambda>0$ and $d\in\RR^\nu$ be given. If $Uv=v$ (so $U\vf_j=\pm\vf_j$ can be arranged for all $j\geq1$), then for any $j\geq1$, $\ve>0$ there exists some $C_\ve(\lambda)>0$ such that \eql{\label{eq:lower bound on rho}
        \abs{\rho(d)} \geq C_\ve(\lambda)\sqrt{-e_j}\exp\br{-d \sqrt{-e_j+\ve}}
        }
    \end{thm}
    \begin{proof}
        We perform a Fourier transform of $\vf_j$ on all coordinates \emph{but} the first one, the result of which we name \eq{\widehat{\vf_j}(x_1,p_\perp)\equiv \int_{x_\perp\in\RR^{\nu-1}}\exp\br{-2\pi\ii x_\perp\cdot p_\perp}\vf_j(x_1,x_\perp)\dif{x_\perp}\,.
        } Schr\"odinger's equation for $\vf_j$, outside of $\supp(v)$, implies for $\widehat{\vf_j}$ the ODE \eq{
        \br{-\partial_1^2 + \norm{p_\perp}^2} \widehat{\vf_j}(\cdot,p_\perp) = e_j \widehat{\vf_j}(\cdot,p_\perp)\qquad(p_\perp\in\RR^{\nu-1})} whose $L^2$ solution is given by \eql{\label{eq:phi hat in terms of A}
        \widehat{\vf_j}(x_1,p_\perp) = A(p_\perp) \exp\br{-\abs{x_1}\sqrt{\norm{p_\perp}^2-e_j}}\qquad(x_1\in\RR)
        } for some function $A$ such that $\widehat{\vf_j}\in L^2(\RR^\nu)$; in writing this expression we have used the fact $\widehat{\vf_j}$ is $L^2$ at $x_1=\pm\infty$ as well as the hypothesis on reflection symmetry. This now yields the expression from \cref{eq:rho with reflection symmetry} \eq{
        \rho(d) = \mp 2 \int_{p_\perp \in \RR^{\nu-1}} \sqrt{\norm{p_\perp}^2-e_j} |A(p_\perp)|^2 \exp\br{-d \sqrt{\norm{p_\perp}^2-e_j}} \dif{p_\perp}\,.
        }

        A theorem of Aronszajn \cite{aronszajn1957unique} implies that $\vf_j$ can't be zero on any open set. If $\vf_j$ vanishes on the hyperplane $x_1=c$, then also $\vf_j=0$ for $x_1>c$, contradicting Aronszajn. Therefore $\vf_j$ can't vanish identically on $x_1=c$. Hence the partial fourier transform $\widehat{\vf_j}$ can't be identically zero on $x_1=c$. Agmon's estimate (see \cref{thm:Agmon's estimate} below) implies that $\vf_j$ is exponentially decreasing in all its arguments, outside of the support of $a$ (which the hyperplane $x_1=c$ is) hence $\widehat{\vf_j}(c,\cdot)$ is real-analytic. As a real-analytic function which is not identically zero, $\widehat{\vf_j}(c,\cdot)$ has $L^2(B_R(0))$ norm non-zero for any $R$. By \cref{eq:phi hat in terms of A}, the same is true for $A$, i.e., for any $R>0$ there exists some $C_R>0$ such that \eq{
        \int_{p_\perp \in B_R(0)} |A(p_\perp)|^2\dif{p_\perp}\geq C_R > 0
        } which implies 
        \eq{
        \abs{\rho(d)} \geq 2 C_R \sqrt{-e_j}\exp\br{-d\sqrt{R^2-e_j}}
        } and hence the result.
        
    \end{proof}

    \section{Connecting the hopping to the non-magnetic energy splitting for excited states}\label{app:excited}
    In this section, \emph{without}  the assumption of reflection symmetry on $v$ we want to establish \cref{eq:connection between rho and Delta}. This differs from \cite[Eq-n (1.12)]{FSW_22_doi:10.1137/21M1429412} in three different ways:
    \begin{enumerate}
        \item The current statement is a non-magnetic result (though our present argument will just as well apply to the magnetic case).
        \item The present statement applies for excited states.
        \item The present statement takes $\lambda$, the depth of the potentials,  sufficiently large but fixed, and  $d\to\infty$.
    \end{enumerate}

    \begin{thm}\label{thm:d_limit}
        Let $j\geq1$ be given. 
        Assume that 
        \begin{enumerate}
            \item $e_j$ is well-separated from the rest of the spectrum of $h$, i.e., we assume that there exists some $\gamma_j>0$ such that \eq{\dist(\sigma(h)\setminus\Set{e_j},e_j)\geq\gamma_j\,.}
            \item We have the relation \eq{
            |e_j|>\gamma_j\,.
            }
            \item $e_j$ is a simple eigenvalue.
        \end{enumerate}
        
        Then, if $d\mapsto \rho_j(\lambda,d)$ obeys a lower bound of the form \cref{eq:lower bound on rho}, then in the notation of the introduction, \eql{
        \lim_{d\to\infty}\frac{\Delta_j(\lambda,d)}{2|\rho_j(\lambda,d)|}=1\qquad(j\geq1)\,.
        }
    \end{thm}

\cref{thm:d_limit} easily implies a lower bound on the eigenvalue splitting in the case where $v$ is reflection symmetric.
\begin{cor} Assume (i) reflection symmetry for $v$ and (ii) simplicity of $e_j$. For fixed $\varepsilon>0$ there exists $d_\ve>0$ such that for all $d\ge d_\ve$
    \eql{
    \Delta_j(\lambda,d) \geq \br{1-\ve}C_\ve(\lambda)\sqrt{-e_j}\exp\br{-d \sqrt{-e_j+\ve}}\,.
    }
\end{cor}
\begin{proof}
For fixed $\varepsilon>0$, set $R=\varepsilon$ in \cref{thm:lb_rho_j}.  
 By the assumptions of (i) reflection symmetry for $v$ and (ii) simplicity of $e_j$,
the lower bound \cref{eq:lower bound on rho} on
$|\rho_j(\lambda,d)|$ holds. Applying \cref{thm:d_limit}, yields the result.
\end{proof}
    
 For single well Hamiltonians of the form $h=P^2+\lambda^2v(X)$, the  first two assumptions of \cref{thm:d_limit} hold if $v$ is smooth and has, for all $\lambda$ sufficiently large, a unique non-degenerate minimum. The simplicity of the eigenvalue $e_j$ is  not guaranteed and a more detailed analysis is necessary. Suppose $e_j$ is degenerate of order  $N\ge2$, with an orthonormal eigenbasis $\big\{\vf_j^1\,\cdots,\vf_j^N\big\}$. One would,
     in the Schur complement reduction (see the proof sketch just below), need to study the $2N\times 2N$ matrix whose matrix elements are \eq{
        \ip{R^n\vf_j^k}{\br{H-e_j\Id}R^m\vf_j^l}
    } for $n,m\in\Set{0,d}$ and $k,l=1,\cdots,N$.

    \begin{proof}[Proof of \cref{thm:d_limit}]
    
        We shall only present a sketch; most of the components of the proof have already been presented in \cite{FSW_22_doi:10.1137/21M1429412}. The proof is comprised of several parts:
        \begin{enumerate}
        \item For $d$ large, we expect that the two vectors $\vf_j, R^d\vf_j$ will be a good approximate basis of the spectral subspace of $H_{\lambda,d}$ near energy $e_j$, since $\br{H-e_j}R^m\vf_j\approx0$ for $m=0,d$.
        \item Via a Schur complement argument, for energies in a neighborhood of  $e_j$, we reduce the Hamiltonian $H_{\lambda,d}-e_j\Id-\Omega \Id$, to a   two by two perturbation of the matrix 
        \begin{equation}\label{eq:2by2}
        \Big(\ \ip{R^n\vf_j}{\br{H_{\lambda,d}-e_j\Id}R^m\vf_j}\Big)_{n,m=0,d} \ -\ \Omega\Id, 
     \end{equation}
     which is analytic in $\Omega$, for $\Omega$ small.
    Formally, this reduction step is completely general and algebraic.
 However, rigorously implementing this reduction requires bounds on the norm of the resolvent of $H_{\lambda,d}$, for energies near $e_j$, on the subspace orthogonal to $\big\{\vf_j, R^d\vf_j\big\}$. A general strategy for such ``energy estimates'', based on \cite[Sections 3 and 4]{FSW_22_doi:10.1137/21M1429412},  is presented in \cref{sec:gen-strat} and is implemented in the setting of \cref{thm:d_limit} in \cref{subsec:special case of energy estimates} below. 
            \item Taking $d\gg1$ and expanding \eqref{eq:2by2}, we find that  eigenvalues $\Omega$ near $0$, correspond to 
            values of $\Omega$ near $0$ for 
            which  
            \[ \rho_j(d)\sigma_1 - \Omega \sigma_0\ +\  {\rm Corrections}(\Omega,d)\approx0,\quad |\Omega|<c.\]
            \item Write $\Omega=\rho_j(d)z$ and apply Rouché's theorem to prove, for $d\gg1$, that the point spectrum of $H_{\lambda,d}$ near $e_j$
is given by two distinct energies, approximately equal to  $e_j \pm |\rho_j(d)|$. Hence, $\Delta_j(d)\approx 2|\rho_j(d)|$ for $d\gg1$. 
\item The application of Rouch\'e's theorem requires that as $d\to\infty$
\begin{equation} \frac{\textrm{Corrections}(\Omega,d)}{|\rho_j(d)|}\ \to\ 0, \quad 
|\Omega|<c.\label{eq:boundCorrections}
\end{equation}
        \end{enumerate}

        We comment briefly on the estimates \eqref{eq:boundCorrections}; cf. \cite[Appendix A]{FSW_22_doi:10.1137/21M1429412}. The analysis greatly simplifies when we take $d\to\infty$ instead of $\lambda$. For instance, consider the error term of the diagonal matrix elements in the $2\times 2$ Schur matrix. The quantity to be bounded, which arises from the $(1,1)$ matrix element, is \eq{
        \frac{\abs{\lambda^2\ip{\vf_j}{ (R^d v)\vf_j}}}{\abs{\rho_j(\lambda,d)}}\,.
        } 
        Using \cref{eq:Agmon} we estimate \eq{
        \abs{\ip{\vf_j}{ (R^d v)\vf_j}} \leq C\exp\br{-2\sqrt{-e_j}\br{d-a}}
        } for some constant $C$ independent on $d$. On the other hand comparing this with \cref{eq:lower bound on rho} we get for any $\vf>0$ \eq{
        \abs{\rho(d)} \geq C_\ve\exp\br{- \sqrt{-e_j+\ve}\, d}
        } and so the entire fraction can easily be made to converge to zero as $d\to\infty$.

        The next term we study is $\ip{\vf_j}{R^d\vf_j}$. That integral is handled by dividing $\RR^\nu$ into three regions, $B_a(0),B_a(d)$ and the complement of both, call it $C$. On each of these sets we have at least one of the eigenfunctions decaying as $\exp\br{-\sqrt{-e_j}d}$. When we consider the orthonormalized translated wave function\eq{
        \widetilde{R^d\vf_j} := \frac{R^d\vf_j-\ip{\vf_j}{R^d\vf_j}\vf_j}{1-\abs{\ip{\vf_j}{R^d\vf_j}}^2}
        } and its associated diagonal overlap integral \eq{
        \ip{\widetilde{R^d\vf_j}}{\br{H-e\Id}\widetilde{R^d\vf_j}}
        } we may control its decay when divided by $\rho_j$.
        
        All other terms follow a similar pattern, because they inherently involve the competition between two factors of $\vf_j$ evaluated at $d$ in the numerator versus one in the denominator.
    \end{proof}

\section{Energy estimates}\label{sec:energy-estimates}

 The strategy outlined in our proof of \cref{thm:d_limit} is based on a (Schur complement) reduction to a matrix problem which relies on an upper bound on the norm of the resolvent of $H-e_j\Id$ when projected to the subspace perpendicular to the span of $\vf_j,R^d\vf_j$. In this section we present a generalization of this bound following the strategy, used for example in \cite{FLW17_doi:10.1002/cpa.21735} and \cite{ShapWein22}. Since this argument appeared at least twice already, and since we anticipate future use in other settings, we elected to present a somewhat general form of it before the application to the present context.


 \subsection{General energy estimates}\label{sec:gen-strat}
 In this section we present a somewhat more general energy estimate that will hopefully be useful also beyond this note. It generalizes e.g. \cite[Section 5]{ShapWein22} by: (1) not concentrating necessarily on the ground state of the one-well system, (2) allowing for degeneracies, (3) taking more general translations and kinetic energy which only need to obey some estimates. It moreover changes the asymptotic parameter from $\lambda$ to the minimal lattice spacing (although that change makes the proof somewhat easier). In the subsequent section we will apply these estimates to our setting of the excited states of a double-well system.
 
Let $T$ be a self-adjoint operator on $L^2(\RR^\nu)$. In principle we anticipate the hypotheses on $T$ stipulated below to hold for all polynomials in $P\equiv-\ii\nabla$, certainly to $P^2$ and to the Landau Hamiltonian. We do not exclude possible extensions to non-local kinetic terms such as $\sqrt{P^2}$ or even general pseudo-differential symbols, but exploring, e.g., for which measurable functions $f$, such that $T=f(P^2)$, the hypothesis below on $T$ fail is beyond the present scope. Let $v:\RR^\nu\to\RR$ be bounded and compactly supported within some bounded open set $S\subseteq \RR^\nu$. We shall also refer to a certain fattening of $S$, say $S^+\supsetneq S$, which shall be another bounded open set. Note that $S^+$ may well depend on $d$ though we keep that implicit in the notation for the time being.

Define 
\eq{
	h = T + v(X)\ .
}
 Since $d$ is our asymptotic parameter and $\lambda$ is fixed, we have absorbed the factor of $\lambda$ in the potential.

Assume that for some $e\in\RR$ there exists some orthonormal set $\Set{\vf_i}_{i=1}^N$ such that $h \vf_i = e \vf_i$ for all $i=1,\cdots,N$. Furthermore, assume that there exists some $\gamma>0$ such that  \eql{\label{eq:atomic gap hypothesis} \norm{\br{h-e\Id}\psi} \geq \gamma \norm{\psi}\qquad\br{\psi\perp \vf_i\forall i=1,\cdots,N}\,.
}		

Let $\GG_d\subseteq\RR^\nu$ be a countable set depending on a parameter $d>0$. In what follows we write $\GG\equiv\GG_d$, and if a statement which involves $\GG$ does not have $d$ explicitly appear we mean the statement holds for any $d>0$ sufficiently large. Let $R:\GG\to\mathcal{U}(L^2(\RR^\nu))$ be given, i.e., $R^n$ is unitary for each $n\in\GG$. We assume it is chosen in such a way that if $f,g\in L^2$ are such that their supports are contained within $S^+$ then \eql{\label{eq:disjoint support assumption}\br{R^n T^\alpha f} R^m T^\beta g = 0} for all $n\neq m$,$\alpha,\beta=0,1,2$. On any operator $A$, we denote \eq{
A_n := \br{R^n}^\ast A R^n \qquad (n\in\GG)
\,.} 

We assume that $[R^n,T]=0$ for all $n\in\GG$ and set $v_n := \br{v(X)}_n$ for all $n\in\GG$.

We define $\calV_n := \szpan\br{\Set{R^n\vf_i}_{i=1}^N}$ and $\calV := \szpan\br{\Set{R^n\vf_i}_{n\in\GG,i=1,\cdots,N}}$ and we let $\Pi_n$ and $\Pi$ respectively be the orthogonal projections onto these spaces.

\begin{defn}\label{def:R-translation of cptly supported functions}
In what follows we will be interested in $R$-translations of a sequence of compactly-supported functions. To that end, say that $\Set{\eta_n}_{n\in\GG}$ is a sequence of functions such that $\supp(\eta_n)\subseteq S^+$ for each $n\in\GG$. Then by \cref{eq:disjoint support assumption} we have $\br{R^{n'}\eta_n}R^{m'}f = 0 $ for all $n'\neq m'\in\GG$ and for all $n\in\GG$, $f\in L^2$ with $\supp(f)\subseteq S^+$. For such a sequence we then define \eq{
\eta := \sum_{n\in\GG} R^n \eta_n\,.
}
\end{defn}

\begin{assumption}\label{ass:summability of error term}
We assume that if $\eta$ is constructed as in \cref{def:R-translation of cptly supported functions} then there exists some $\ve_d$ such that $\ve_d\to0$ as $d\to\infty$ for which
\eql{\label{eq:summability of error term}
	\sum_{n\in\GG}\norm{\Pi_n \sum_{m\neq n}R^m\eta_m}^2 \leq \ve_d^2\sum_{m\in\GG}\norm{\eta_m}^2\,.
}
\end{assumption}
\begin{thm}
	Let \eq{
	H := T + \sum_{n\in\GG} v_n\,.
	}
	
	For any $\psi\in\calV^\perp$ and $d$ sufficiently large there exists some $\xi_d\in(0,\infty)$ such that $\xi_d\to0$ as $d\to\infty$ such that \eql{
	\norm{\br{H-e\Id}\psi} \geq \br{\gamma - \xi_d }\norm{\psi}\,.
	}
\end{thm}

To prove the theorem we present a sequence of lemmas that build $\psi\in\calV^\perp$ gradually.

\begin{lem}[disjoint energy estimate]\label{lem:first energy estimate}
	Let $\psi\in\calV^\perp$ such that additionally, $\psi=\sum_n R^n \eta_n$, i.e., it is constructed as in \cref{def:R-translation of cptly supported functions}. 
	
	Then \eql{
	\norm{\br{H-e\Id}\psi} \geq \sqrt{\gamma^2-\ve_d^2}\norm{\psi}
	}
\end{lem}
\begin{proof}
	For convenience within this proof we shift $\eta_n \to \ti \eta_n$ and then define $\eta_n := R^n \ti \eta_n$.
	
	We have \eq{
	\norm{\br{H-e\Id}\psi}^2 &= \ip{\psi}{\br{H-e\Id}^2\psi}\\
	&=\sum_{n,m\in\GG}\ip{\br{H-e\Id}\eta_n}{\br{H-e\Id}\eta_m}	\,.
	}
	Let us now write $ H = h_n + \sum_{m\neq n } v_m$ where $h_n \equiv T + v_n$. By hypothesis we have $\eta_m v_n = 0$ for all $n\neq m$ so that \eq{
	\br{H-e\Id}\eta_m = \br{h_m-e\Id}\eta_m\,.
	} We thus find \eq{
	\norm{\br{H-e\Id}\psi}^2 &= \sum_{n,m\in\GG}\ip{\br{h_n-e\Id}\eta_n}{\br{h_m-e\Id}\eta_m} \\ 
	&= \sum_{n\in\GG}\ip{\br{h_n-e\Id}\eta_n}{\br{h_n-e\Id}\eta_n} + \sum_{n\neq m\in\GG}\ip{\br{h_n-e\Id}\eta_n}{\br{h_m-e\Id}\eta_m}\,.
	} Now, for $n\neq m$, we have \eq{
	\ip{\br{h_n-e\Id}\eta_n}{\br{h_m-e\Id}\eta_m} = \ip{h_n\eta_n}{h_m\eta_m} + e^2 \ip{\eta_n}{\eta_m} - e\ip{h_n \eta_n}{\eta_m}-e\ip{\eta_n}{h_m \eta_m} \,.
	} Using \cref{eq:disjoint support assumption} every term in the above line is zero. We thus find \eq{
	\norm{\br{H-e\Id}\psi}^2 &= \sum_{n\in\GG}\norm{\br{h_n-e\Id}\eta_n}^2\,.
	}
	
	Let us write now $\eta_n := \Pi_n \eta_n + \Pi_n^\perp \eta_n$. We emphasize that despite $\psi\in\calV^\perp$, we cannot assume $\eta_n \in\calV_n$ separately. By definition, \eq{
	\br{h_n-e\Id}\eta_n = \br{h_n-e\Id}\Pi_n^\perp\eta_n\,.
	} Using the gap assumption \cref{eq:atomic gap hypothesis} we have then $\norm{\br{h_n-e\Id}\Pi_n^\perp\eta_n} \geq \gamma \norm{\Pi_n^\perp \eta_n}$ so that \eq{
	\norm{\br{H-e\Id}\psi}^2 \geq \gamma^2\sum_{n\in\GG}\norm{\Pi_n^\perp\eta_n}^2=\gamma^2\sum_{n\in\GG}\br{\norm{\eta_n}^2-\norm{\Pi_n\eta_n}^2}=\gamma^2\br{\norm{\psi}^2-\sum_{n\in\GG}\norm{\Pi_n\eta_n}^2}\,.
	}
	We now bound $\sum_{n\in\GG}\norm{\Pi_n\eta_n}^2$. Since $\psi\in\calV^\perp$ it is also in $\calV_n^\perp$ and so \eq{
		\Pi_n\eta_n = \Pi_n\br{\psi-\sum_{m\neq n}\eta_m} = -\sum_{m\neq n}\Pi_n\eta_m} at which point \cref{eq:summability of error term} applies, and so we obtain
		\eq{
		\norm{\br{H-e\Id}\psi}^2 \geq \br{\gamma^2-\ve_d^2}\norm{\psi}^2
		} and hence the result.
\end{proof}

\begin{lem}[localized energy estimate]\label{lem:loc energy estimate}
	Let $\psi\in\calV^\perp$ and $\Theta_d:\RR^\nu\to[0,1]$ be some function such that \eql{\label{eq:assumption on partition of unity}
		\norm{\br{\Theta_d(X)-\Id}\Pi}<\ti\ve_d
	} with $\ti\ve_d\to0$ as $d\to\infty$ and such that for any $f\in L^2$, $\Theta_d f$ is of the form \cref{def:R-translation of cptly supported functions}. Further assume that \eql{\label{eq:commutator of Hamiltonian with Theta is bounded}
	\sup_{d>0,\alpha=1,2}\norm{[\br{H-e\Id}^\alpha,\Theta_d]\Pi} < \infty\,.
	} and \eql{\label{eq:Hamiltonian on orbitals is bounded}
    \sup_{d>0,\alpha=1,2}\norm{\br{H-e\Id}^\alpha\Pi}<\infty\,.
    } 
    
    Then for any $\alpha\in(0,1)$ there is a constant $M_{d,\alpha}$ such that $\sup_{d>0}M_{d,\alpha}<\infty$ and \eq{
	\norm{\br{H-e\Id}\Theta\psi}^2 &\geq \br{\gamma^2-\ve_d^2}\alpha\norm{\Theta\psi}^2-2M_{d,\alpha}\br{1-\ti\ve_d}^{-1}\ti\ve_d\norm{\psi}^2
\,.	}
\end{lem}
\begin{proof}
	Even though we allow $\Theta_d$ to depend on $d$ to avoid clutter in this proof we merely write $\Theta$.
	
	 The assumption \cref{eq:assumption on partition of unity} implies that the operator $\left.\Pi \Theta(X)\Pi\right|_{\calV}$ is invertible. Indeed, 
	 \eq{
	 	\left.\Pi \Theta(X)\Pi\right|_{\calV} = \left.\Id\right|_\calV + \left.\Pi \br{\Theta(X)-\Id}\Pi\right|_\calV
	 } and $\Id+A$ is invertible if $\norm{A}<1$. With this, let us define \eq{
	 u  := \br{\left.\Pi \Theta(X)\Pi\right|_{\calV}}^{-1} \Pi \Theta(X)\Pi^\perp \psi
	 } and \eq{
	 \eta := \Theta(X) \br{\psi - u }\,.} We have chosen $u $ so as to correct for the fact that even though $\Pi \psi = 0$, $\Pi\Theta(X)\psi \neq 0$. However, now that we have $\eta$, we have 
	 \eq{
		 \Pi \eta = \Pi \Theta(X) \br{\psi-u } = \Pi \Theta(X) \psi - \Pi \Theta(X) \br{\left.\Pi \Theta(X)\Pi\right|_{\calV}}^{-1} \Pi \Theta(X)\Pi^\perp \psi = 0 \,,
	 } i.e., $\eta\in\calV^\perp$. This latter fact together with the assumption \cref{eq:decomposition assumption on Theta} implies that we may invoke \cref{lem:first energy estimate} to get \eq{
	 \norm{\br{H-e\Id}\eta} \geq \sqrt{\gamma^2-\ve_d^2}\norm{\eta}\,.
	 } 
	 
	 We now obtain a lower bound on $\norm{\eta}$. Using $\Pi \Theta(X)\Pi^\perp = \Pi\br{\Theta-\Id}+\Pi\br{\Id-\Theta}\Pi$ we have \eq{
	 \norm{u} &\leq \norm{\br{\left.\Pi \Theta(X)\Pi\right|_{\calV}}^{-1} \Pi \Theta(X)\Pi^\perp \psi} \\
	 &\leq \br{1-\ti\ve_d}^{-1}\norm{\Pi \Theta(X)\Pi^\perp}\norm{\psi}\\
	 &\leq \br{1-\ti\ve_d}^{-1}2\ti\ve_d\norm{\psi}=:\delta_d \norm{\psi}\,.
	 } We have again $\delta_d\to0$ as $d\to\infty$. Hence, for any $\alpha\in(0,1)$,
	 \eq{
	 \norm{\eta}^2 &\geq \br{\norm{\Theta\psi}-\norm{\Theta u}}^2 \\
	 &\geq \alpha\norm{\Theta\psi}^2 - \frac{\alpha}{1-\alpha}\norm{\Theta u}^2\\
	 &\geq \alpha\norm{\Theta\psi}^2 - 4\frac{\alpha}{1-\alpha}\delta_d^2\norm{\psi}^2\,.
	 }
	 
	 Let us write 
	 \eq{
	 \norm{\br{H-e\Id}\eta}^2 &= \ip{\eta}{\br{H-e\Id}^2\eta} \\
	 &= \ip{\Theta \psi}{\br{H-e\Id}^2\Theta\psi} + \ip{\Theta u }{\br{H-e\Id}^2 \Theta u} -2 \Re{\ip{\Theta \psi}{\br{H-e\Id}^2\Theta u}}\,.
	 }
	 
	 Putting all these estimates together we find 
	 \eq{
	 \norm{\br{H-e\Id}\Theta\psi}^2 &\geq \br{\gamma^2-\ve_d^2}\alpha\norm{\Theta\psi}^2-\\
	 &\qquad -\br{\gamma^2-\ve_d^2} \frac{\alpha}{1-\alpha}\delta_d^2\norm{\psi}^2-\ip{\Theta u }{\br{H-e\Id}^2 \Theta u}+\\
	 &\qquad +2 \Re{\ip{\Theta \psi}{\br{H-e\Id}^2\Theta u}}\,.
	 }
	 We thus need to find a bound on the latter two terms on the right hand side.
	 
	 For the first term, since $u\in\calV$, we have 
	 \eq{
	 \norm{\br{H-e\Id}\Theta u} &\leq \norm{\Theta \br{H-e\Id}u} + \norm{[H-e\Id,\Theta]u} \\
	 &\leq \norm{\br{H-e\Id}\Pi}\norm{u} + \norm{[T,\Theta]\Pi}\norm{u}\\
	 &\leq \br{\norm{\br{H-e\Id}\Pi}+\norm{[T,\Theta]\Pi}}\delta_d\norm{\psi}\,.
	 }
	 
	 For the last term, using $u\in\calV$ again, 
	 \eq{
	 \abs{\ip{\Theta \psi}{\br{H-e\Id}^2\Theta u}} &\leq \norm{\psi}\norm{\br{H-e\Id}^2\Theta u} \\
	 &\leq \norm{\psi}\br{\norm{\br{H-e\Id}^2\Pi}+\norm{[\br{H-e\Id}^2,\Theta]\Pi}}\norm{u}\\
	 &\leq \br{\norm{\br{H-e\Id}^2\Pi}+\norm{[\br{H-e\Id}^2,\Theta]\Pi}}\delta_d\norm{\psi}^2\,.
	 }
	 
	 Collecting everything together we find
	 \eq{
	 	\norm{\br{H-e\Id}\Theta\psi}^2 &\geq \br{\gamma^2-\ve_d^2}\alpha\norm{\Theta\psi}^2-M\delta_d\norm{\psi}^2\
	 } with 
	 \eq{ 
	 	M:=\br{\gamma^2-\ve_d^2}\frac{\alpha}{1-\alpha}\delta_d+\br{\norm{\br{H-e\Id}\Pi}+\norm{[T,\Theta]\Pi}}^2\delta_d+2\norm{\br{H-e\Id}^2\Pi}+2\norm{[\br{H-e\Id}^2,\Theta]\Pi}\,.
	}
    Now \cref{eq:commutator of Hamiltonian with Theta is bounded,eq:Hamiltonian on orbitals is bounded} guarantee that $M$ is bounded as a function of $d$, so we are finished.
\end{proof}

\begin{lem}\label{lem:global energy estimate} Assume there is a $\Theta_d:\RR^\nu\to[0,1]$ such that \cref{lem:loc energy estimate} holds, and such that, further, $\Sigma_d:=\sqrt{1-\Theta_d^2}$ has $\Sigma_d v_n=0$ for all $n\in\GG$.

We further assume $\Theta_d$ has been chosen such that there are constants $A_d,B_d,C_d\in(0,\infty)$, all converge to zero as $d\to\infty$, such that \eql{\label{eq:the decay of the commutators}
\max_{f=\Theta_d,\Sigma_d}\abs{\ip{f(X)\psi}{[f(X),\br{H-e\Id}^2]\psi}} \leq A_d\norm{\psi}^2 + B_d\norm{\br{H-e\Id}\psi}^2 \,.
} 

If $T\geq0$, $e\leq0$, and $|e| > \gamma$ then for all $\psi\in\calV^\perp$ we have \eql{
    \norm{\br{H-e\Id}\psi}^2\geq \frac{\br{\gamma^2-\ve_d^2}\alpha-M_{d,\alpha}\delta_d-2A_d}{1-2B_d}\norm{\psi}^2\,.
    }
\end{lem}
\begin{proof}
    Let $\Theta_d:\RR^\nu\to[0,1]$ be given and assume it is such that the conclusion of \cref{lem:loc energy estimate} holds. As before we write $\Theta\equiv\Theta_d$ and also set $\Sigma := \sqrt{\Id-\Theta^2}$. Let also $\psi\in\calV^\perp$ be given.

    We start with a preliminary 
    \eq{
    \norm{\br{H-e\Id}\psi}^2 &= \ip{\psi}{\br{H-e\Id}^2\psi} \\
    &= \Re{\ip{\Theta^2\psi}{\br{H-e\Id}^2\psi} + \ip{\Sigma^2\psi}{\br{H-e\Id}^2\psi}} \\
    &= \ip{\Theta\psi}{\br{H-e\Id}^2\Theta\psi} + \ip{\Sigma\psi}{\br{H-e\Id}^2\Sigma\psi} + \\
    & + \Re{\ip{\Theta\psi}{[\Theta,\br{H-e\Id}^2]\psi}} + \Re{\ip{\Sigma\psi}{[\Sigma,\br{H-e\Id}^2]\psi}}
    }

    By \cref{lem:loc energy estimate} we have (with $\delta_d$ as defined in its proof) \eq{
	\norm{\br{H-e\Id}\Theta\psi}^2 &\geq \br{\gamma^2-\ve_d^2}\alpha\norm{\Theta\psi}^2-M_{d,\alpha}\delta_d\norm{\psi}^2 \\ 
    &= \br{\gamma^2-\ve_d^2}\alpha\norm{\psi}^2-\br{\gamma^2-\ve_d^2}\alpha\norm{\Sigma\psi}^2-M_{d,\alpha}\delta_d\norm{\psi}^2\,.
	} 
    Moreover, since $\Sigma v_n=0$ for any $n\in\GG$ and since we are assuminug $-e T\geq0$ we get
    \eq{
    \norm{\br{H-e\Id}\Sigma\psi}^2 &= \norm{\br{T-e\Id}\Sigma\psi}^2 \\
    &= \ip{\Sigma\psi}{\br{T^2+e^2\Id-2eT}\Sigma\psi} \\
    &\geq e^2\norm{\Sigma\psi}^2\,.
    }

    Collecting everything together we find
    \eq{
        \norm{\br{H-e\Id}\psi}^2 &\geq \br{\gamma^2-\ve_d^2}\alpha\norm{\psi}^2+\br{e^2-\br{\gamma^2-\ve_d^2}\alpha}\norm{\Sigma\psi}^2-M_{d,\alpha}\delta_d\norm{\psi}^2+ \\
    & + \Re{\ip{\Theta\psi}{[\Theta,\br{H-e\Id}^2]\psi}} + \Re{\ip{\Sigma\psi}{[\Sigma,\br{H-e\Id}^2]\psi}}\\
    &\geq \br{\gamma^2-\ve_d^2}\alpha\norm{\psi}^2-M_{d,\alpha}\delta_d\norm{\psi}^2- \\
    & -2 A_d \norm{\psi}^2-2 B_d \norm{\br{H-e\Id}}^2\,.
    } In the last inequality we used the fact that $e>\gamma$ and that $\ve_d\to0$ as $d\to\infty$, as well as the hypothesis \cref{eq:the decay of the commutators}. This implies the result.
    
\end{proof}
\subsection{Energy estimates for excited states of non-magnetic double-well systems}\label{subsec:special case of energy estimates}
In this section we apply the general scheme of the previous section to the present problem of establishing energy estimates for excited states of non-magnetic double-well systems. To that end we make the following choices w.r.t. the above notation:
\begin{enumerate}
    \item $T$, the kinetic term, shall be $P^2$ with $P\equiv-\ii\nabla$.
    \item If $v$ is negative, bounded, and compactly supported within some $S:=B_a(0)$ for some $a>0$, then any negative-energy bound state must lie in the discrete spectrum and hence be finitely-degenerate, so assuming $e<0$ and $N<\infty$ is reasonable, whence \cref{eq:atomic gap hypothesis} follows. At this point we may \emph{assume} $e > \gamma$, though it is a generic situation whenever $v$ is sufficiently deep (which may be arranged by replacing $v$ with $\lambda^2v$ and taking $\lambda$ sufficiently large; this is done before determining $d$). We let $S^+:=B_{a^+}(0)$ with $a^+>a$. In fact we will choose eventually $a^+$ to depend on $d$, $a^+\equiv a^+(d)$
    
    For convenience we assume the eigenstates are normalized, $\norm{\vf_j}=1$, and name $\kappa:=\sqrt{-e}$ and we assume $\kappa>1$.
    \item For the lattice $\GG_d$ we take the set $\Set{0,d e_1}\subseteq\RR^\nu$, whence  $R:\GG_d\to\calU(L^2(\RR^\nu))$ yields $R^0:=\Id$ and translation by $d$: \eq{
    \br{R^d f}(x) := f(x-de_1)\qquad(f\in L^2\,;\qquad x\in\RR^\nu)\,.
    }
    With these choices \cref{eq:disjoint support assumption} is fulfilled.

\end{enumerate}

	 Now a standard Agmon estimate (\cite{agmon1982lectures}; see \cref{thm:Agmon's estimate} below) implies that if $h \vf = e \vf$ with $e<0$ and $\vf\in L^2$, then \eql{
|\vf(x)| \leq C_\vf \exp\br{-\kappa \norm{x}} \qquad \br{\norm{x}\geq a+1 }
} for some $C_\vf<\infty$.

In particular this means that if $f\in L^2$ has support within $B_r(x)$ with $\norm{x}-r>a+1$ then \eql{
\abs{\ip{f}{\vf}} \leq C_\vf \int_{y\in B_r(x)}\exp\br{-\kappa\norm{y}}\abs{f(y)}\dif{y} \leq C_\vf \sqrt{\abs{B_r(0)}}\norm{f}_{L^2} \exp\br{-\kappa \br{\norm{x}-r}}\,.
}

With these choices, we want to establish that the various hypothesis made in the previous section at each step actually hold. We do so in the order that the hypothesis appear:
\begin{lem}
    \cref{ass:summability of error term} holds with the choices made in this subsection, if $d-a^+(d)\to\infty$ as $d\to\infty$.
\end{lem}
\begin{proof}
    We now have $\Pi_n := \sum_{j=1}^NR^n \vf_j\otimes \vf_j^\ast R^{-n}$ for $n=0,d$. That means we have 
    \eq{
    \sum_{n\in\GG}\norm{\Pi_n \sum_{m\neq n}\eta_m}^2 &= \norm{\Pi_0 R^d \eta_d}^2 + \norm{\Pi_d \eta_0}^2 \\
    &= \sum_{j=1}^N \abs{\ip{\vf_j}{R^d\eta_d}}^2 \norm{\vf_j}^2 + \abs{\ip{R^d\vf_j}{\eta_0}}^2 \norm{R^d\vf_j}^2\\
    &\leq  2N C_{\vf_1}^2 \abs{B_{a^+(d)}(0)} \exp\br{-2\kappa\br{d-{a^+(d)}}}\br{\norm{\eta_d}^2_{L^2}+\norm{\eta_0}^2_{L^2}}\,.
    } We thus identify \eq{
    \ve_d := 2N C_{\vf_1}^2 \abs{B_{a^+(d)}(0)} \exp\br{-2\kappa\br{d-a^+(d)}}
    } which converges to zero. 
\end{proof}

We are now in a position to choose our smooth partition of unity $\Theta_d:\RR^\nu\to[0,1]$. Let $\Theta_d(x) := f(\norm{x})+f(\norm{x-d})$ where $f:[0,\infty)\to[0,1]$ is a  smooth bump function with \eq{
\left.f\right|_{[0,a^-(d)]} &= 1 \\
\left.f\right|_{[a^+(d),\infty)} &= 0\,.
} for some $a^-(d)<a^+(d)$. 

One choice of $a^\pm(d)$ which will do is for example $a^-(d) = \frac13d$ and $a^+(d) = \frac12d$. 

\begin{lem}
    For the above choice of $\Theta_d$ \cref{eq:assumption on partition of unity,eq:commutator of Hamiltonian with Theta is bounded,eq:Hamiltonian on orbitals is bounded} hold.
\end{lem}
\begin{proof}
    We omit the proof of \cref{eq:assumption on partition of unity} which is trivial except for the fact that $\Pi\neq\Pi_0+\Pi_d$, but that fact only adds terms which decay exponentially as $\exp\br{-\kappa d}$. Ignoring that fact, since $\Id-\Theta_d$ is supported away from the balls of radius $r_1$, it is clear that we can guarantee decay of this term at least as fast as \eq{
        \exp\br{-\kappa a^{-}(d)}\,.
    } We thus identify \eq{
    \ti\ve_d := \text{const.}\times\exp\br{-\kappa a^{-}(d)}\,.
    }

    Next, we turn to \cref{eq:Hamiltonian on orbitals is bounded}. We write any $\psi\in\im\Pi$ as $\psi = \sum_{j=1}^N \alpha_j \vf_j+\beta_j R^d\vf_j$ for some coefficients $\alpha_j,\beta_j$. Then \eq{
    \br{H-e\Id} \psi &= \sum_{j=1}^N \alpha_j \alpha_j v_d \vf_j + \beta_j v_0 R^d \vf_j 
    } and so \eq{
    \norm{\br{H-e\Id} \psi} &\leq \text{const.} \times N \norm{\psi} \exp\br{-\kappa d}
    } and similarly for the second power; again we note that going from the coefficients to $\norm{\psi}$ involves the fact that $\Pi_0\Pi_d\neq0$ which only contributes corrections of the form $\exp\br{-\kappa d}$.

    We now study the commutator, towards \cref{eq:commutator of Hamiltonian with Theta is bounded}:
    \eq{
    [H-e\Id,\Theta_d] &= [P^2,\Theta_d]\\
    &= (P^2\Theta_d) + 2 (P\Theta_d) \cdot P\,.
    } The first term is clearly bounded. For the second term, we have for any $\psi\in\calV$, \eq{
    \norm{P_j \psi}^2 &= \ip{\psi}{P^2\psi} \\ 
    &= \ip{\psi}{\br{H-e\Id}\psi} - \ip{\psi}{\br{v_0+v_d-e\Id}\psi}
    } all of which are bounded uniformly in $d$. The commutator with the second power is dealt with in a similar manner.
\end{proof}

With our particular choice of $\Theta_d$, we also have $\br{\Id-\Theta_d}v_n=0$ for $n=0,d$, as required by \cref{lem:global energy estimate}. Moreover, clearly for us, $T\equiv P^2\geq 0$ and $e<0$ by hypothesis. Moreover, $e > \gamma$ is also part of our hypothesis. We are thus only left with establishing \cref{eq:the decay of the commutators}. The idea here is that since we have arbitrarily large distance over which $\Theta_d$ goes from $1$ to $0$, we can make its derivatives arbitrarily small, and hence $[\Theta_d,P]$ can be taken arbitrarily small. Similarly for $[\Sigma_d,P]$ with $\Sigma_d := \sqrt{1-\Theta_d^2}$.

\begin{lem}
    We may arrange for \cref{eq:the decay of the commutators} to hold with a choice of $\Theta_d$ that has arbitarily small derivatives.
\end{lem}
\begin{proof}
    We write $f\equiv f(X)$. First let us calculate 
    \eq{
        [f,\br{H-e\Id}^2] &= \br{H-e\Id}[f,P^2]+[f,P^2]\br{H-e\Id}
    } and we use $[P^2,f]
    = (P^2f) + 2 (Pf) \cdot P$ so that
    \eq{
    \abs{\ip{f\psi}{[f,\br{H-e\Id}^2]\psi}} &\leq \abs{\ip{f\psi}{\br{H-e\Id}(P^2f)\psi}} + 2\abs{\ip{f\psi}{\br{H-e\Id}(Pf)\cdot P\psi}} + \\
    & + \abs{\ip{f\psi}{(P^2f)\br{H-e\Id}\psi}} + 2\abs{\ip{f\psi}{(Pf)\cdot P\br{H-e\Id}\psi}} \,.
    }

    Let us study the first term:
    \eq{
    \abs{\ip{f\psi}{\br{H-e\Id}(P^2f)\psi}} &= \abs{\ip{\br{H-e\Id}f\psi}{(P^2f)\psi}} \\
    &\leq \abs{\ip{f\br{H-e\Id}\psi}{(P^2f)\psi}} + \abs{\ip{[\br{H-e\Id},f]\psi}{(P^2f)\psi}} \\
    &\leq \norm{f}_\infty \norm{P^2 f}_\infty \norm{\br{H-e\Id}\psi}\norm{\psi}+\abs{\ip{\br{P^2 f}\psi}{(P^2f)\psi}} + \\
    &+\abs{\ip{\br{P f}P\psi}{(P^2f)\psi}} \\ 
    &\leq \norm{f}_\infty \norm{P^2 f}_\infty \norm{\br{H-e\Id}\psi}\norm{\psi}+\norm{P^2 f}_\infty^2 \norm{\psi}^2 + \\
    &+\norm{Pf}_\infty \norm{P^2 f}_\infty\norm{P\psi}\norm{\psi}\,. 
    }
    All of the terms in the last expression are of the form we need except the very last one, for which we estimate
    \eq{
    \norm{P \psi}^2 &= \ip{\psi}{P^2\psi} \\
    &\leq \abs{\ip{\psi}{\br{H-e\Id}\psi}} + \abs{\ip{\psi}{\br{V-e\Id}\psi}} \\
    &\leq \norm{\psi}\norm{\br{H-e\Id}\psi} + \br{\abs{e}+\norm{V}_\infty}\norm{\psi}^2\,.
    }

    Finally we use Young's inequality to get 
    \eq{
    \norm{\psi}\norm{\br{H-e\Id}\psi} \leq \frac12 \norm{\psi}^2+\frac12\norm{\br{H-e\Id}\psi}^2\,.
    }

    All other three terms are dealt with similarly. The conclusion is that every single term contains either a factor of $\norm{P f}_\infty$ or $\norm{P^2 f}_\infty$, both of which can be made arbitrarily small, and then some constant times either $\norm{\psi}^2$ or $\norm{\br{H-e\Id}\psi}^2$.
\end{proof}
\appendix
\section{Agmon's estimate} For completeness we cite here Agmon's estimate \cite{agmon1982lectures}. 
\begin{thm}[Agmon's estimate]\label{thm:Agmon's estimate}
    Let $H = P^2 + V(X)$ be given on $L^2(\RR^\nu)$, where $V:\RR^\nu\to\RR$ is of compact support, and, say $L^2$, and assume that $\psi\in L^2$ is such that $H\psi=E\psi$ for some $E<0$. Then $\psi$ exhibits pointwise exponential decay with rate $\sqrt{-E}$ outside of $\supp(V)$. In particular, there exists a constant $C<\infty$ so that if $\norm{x}>C$ then 
    \eql{
    \abs{\psi(x)} \leq C \exp\br{-\sqrt{-E}\norm{x}}\,.
    }
\end{thm} 
\begin{proof}
    Rewriting $H\psi=E\psi$ as \eq{
    \br{P^2 -E\Id}\psi = -V(X)\psi
    } we apply the free resolvent on the LHS to get 
    \eq{
        \psi = -\br{P^2 -E\Id}^{-1}V(X)\psi\,.
    } Now evaluating at some point $x\in\RR^\nu$ which is outside of $\supp(V)$ we find 
    \eq{
    \psi(x) = -\int_{y\in\supp(V)} \br{P^2 -E\Id}^{-1}(x,y)V(y)\psi(y)\dif{y}
    } so that 
    \eq{
    \abs{\psi(x)} \leq \norm{V}_{L^2(\supp(V))}\norm{\psi}_{L^2(\supp(V))}\sup_{y\in\supp(V)}\br{P^2 -E\Id}^{-1}(x,y)\,.
    } The Greens function of the free Laplacian is given by \eql{
    \br{P^2 -E\Id}^{-1}(x,y) = \frac{1}{(2\pi)^{\nu/2}} \left( \frac{\sqrt{-E}}{2} \right)^{\nu/2 - 1}  \frac{K_{\nu/2 - 1}(\sqrt{-E} \norm{x-y})}{\norm{x-y}^{\nu/2 - 1}}
    } where $K_m$ is the modified Bessel function of order $m$. 

    In particular let $a>0$ be such that $\supp(V)\subseteq B_a(0)$. Then, using known decay estimates of $K$, with $\kappa := \sqrt{-E}$, we have for $\norm{x}\geq a+\frac1\kappa$, 
    \eql{\label{eq:Agmon}
    |\psi(x)|\leq C_{V,a,\kappa,\psi} \frac{1}{\br{\norm{x}-a}^{\frac{\nu-1}{2}}}\exp\br{-\kappa\br{\norm{x}-a}}
    } with \eq{C_{V,a,\kappa,\psi} &:= \frac{1}{\br{2\pi}^{\nu/2}}\sqrt{\frac{\pi}{2}}\kappa^{\frac{\nu-2}{2}}\norm{V}_{L^2(\supp(V))}\norm{\psi}_{L^2(\supp(V))}\,.}
\end{proof}

		\begingroup
		\let\itshape\upshape
		\printbibliography

@book{Ashcroft_Mermin_1976,
	Author = {Neil W. Ashcroft and N. David Mermin},
	Title = {Solid State Physics},
	Publisher = {Brooks Cole},
	Year = {1976},
	ISBN = {0030839939},
	URL = {http://www.amazon.com/Solid-State-Physics-Neil-Ashcroft/dp/0030839939%3FSubscriptionId%3D0JYN1NVW651KCA56C102%26tag%3Dtechkie-20%26linkCode%3Dxm2%26camp%3D2025%26creative%3D165953%26creativeASIN%3D0030839939}
}

@article{Simon_1984_10.2307/2007072,
	ISSN = {0003486X},
	URL = {http://www.jstor.org/stable/2007072},
	author = {Barry Simon},
	journal = {Annals of Mathematics},
	number = {1},
	pages = {89--118},
	publisher = {Annals of Mathematics},
	title = {Semiclassical Analysis of Low Lying Eigenvalues, II. Tunneling},
	volume = {120},
	year = {1984}
}

@book{landau1977quantum,
  author    = {L. D. Landau and E. M. Lifshitz},
  title     = {Quantum Mechanics: Non-Relativistic Theory},
  series    = {Course of Theoretical Physics},
  volume    = {3},
  edition   = {3rd},
  publisher = {Pergamon Press},
  year      = {1977},
  address   = {Oxford},
  note      = {See \S 50, Problem 3, for discussion of double-well potential and tunneling-induced energy level splitting},
  isbn      = {9780750635394}
}

@article{Helffer_Sjostrand_1984,
	author = {   B.   Helffer  and    J.   Sjostrand },
	title = {Multiple wells in the semi-classical limit I},
	journal = {Communications in Partial Differential Equations},
	volume = {9},
	number = {4},
	pages = {337-408},
	year  = {1984},
	publisher = {Taylor & Francis},
	doi = {10.1080/03605308408820335},
	
	URL = { 
	https://doi.org/10.1080/03605308408820335
	
	},
	eprint = { 
	https://doi.org/10.1080/03605308408820335
	
	}
	
}

@article{FLW17_doi:10.1002/cpa.21735,
	author = {Fefferman, Charles L. and Lee-Thorp, James P. and Weinstein, Michael I.},
	title = {Honeycomb Schr\"odinger Operators in the Strong Binding Regime},
	journal = {Communications on Pure and Applied Mathematics},
	
	volume = {71},
	number = {6},
	pages = {1178-1270},
	doi = {10.1002/cpa.21735},
	url = {https://onlinelibrary.wiley.com/doi/abs/10.1002/cpa.21735},
	eprint = {https://onlinelibrary.wiley.com/doi/pdf/10.1002/cpa.21735},
	year = {2018}
}

@article{ShapWein22,
title = {Tight-binding reduction and topological equivalence in strong magnetic fields},
journal = {Advances in Mathematics},
volume = {403},
pages = {108343},
year = {2022},
issn = {0001-8708},
doi = {https://doi.org/10.1016/j.aim.2022.108343},
url = {https://www.sciencedirect.com/science/article/pii/S0001870822001591},
author = {Jacob Shapiro and Michael I. Weinstein}
}

@article{FSW_22_doi:10.1137/21M1429412,
author = {Fefferman, Charles and Shapiro, Jacob and Weinstein, Michael I.},
title = {Lower Bound on Quantum Tunneling for Strong Magnetic Fields},
journal = {SIAM Journal on Mathematical Analysis},
volume = {54},
number = {1},
pages = {1105-1130},
year = {2022},
doi = {10.1137/21M1429412},

URL = { 
    
        https://doi.org/10.1137/21M1429412
    
    

},
eprint = { 
    
        https://doi.org/10.1137/21M1429412
    
    

}
}

@article{Simon_1983_AIHPA_1983__38_3_295_0,
     author = {Simon, Barry},
     title = {Semiclassical analysis of low lying eigenvalues. {I.} {Non-degenerate} minima : asymptotic expansions},
     journal = {Annales de l'I.H.P. Physique th\'eorique},
     pages = {295--308},
     publisher = {Gauthier-Villars},
     volume = {38},
     number = {3},
     year = {1983},
     zbl = {0526.35027},
     mrnumber = {708966},
     language = {en},
     url = {http://www.numdam.org/item/AIHPA_1983__38_3_295_0/}
}

@BOOK{Dimassi2010-hq,
  title     = "London mathematical society lecture note series: Spectral
               asymptotics in the semi-classical limit series number 268",
  author    = "Dimassi, M and Sjostrand, J",
  publisher = "Cambridge University Press",
  month     =  mar,
  year      =  2010,
  address   = "Cambridge, England"
}

@article{aronszajn1957unique,
  author    = {Aronszajn, Nachman},
  title     = {A unique continuation theorem for solutions of elliptic partial differential equations of second order},
  journal   = {Journal de Mathématiques Pures et Appliquées},
  volume    = {36},
  pages     = {235--249},
  year      = {1957}
}

@article{fournais2025purely,
  author    = {Søren Fournais and Léo Morin and Nicolas Raymond},
  title     = {Purely magnetic tunneling between radial magnetic wells},
  journal   = {Revista Matemática Iberoamericana},
  year      = {2025},
  doi       = {10.4171/RMI/1526},
  note      = {Published online March 5, 2025},
  url       = {https://ems.press/journals/rmi/articles/14298628}
}

@book{agmon1982lectures,
  author    = {Shmuel Agmon},
  title     = {Lectures on Exponential Decay of Solutions of Second-Order Elliptic Equations: Bounds on Eigenfunctions of N-Body Schrödinger Operators},
  series    = {Mathematical Notes},
  volume    = {29},
  publisher = {Princeton University Press},
  year      = {1982},
  address   = {Princeton, NJ},
  isbn      = {978-0691083016}
}

@article{SlaterKoster1954,
  author       = {J. C. Slater and G. F. Koster},
  title        = {Simplified LCAO Method for the Periodic Potential Problem},
  journal      = {Physical Review},
  year         = {1954},
  volume       = {94},
  number       = {6},
  pages        = {1498--1524},
  doi          = {10.1103/PhysRev.94.1498},
  publisher    = {American Physical Society}
}
		\endgroup
	\end{document}